\documentclass{article}                  
\usepackage{graphicx}

\usepackage{amsmath,amsfonts,amssymb,amsthm}
\usepackage{natbib}
\usepackage[sans]{dsfont}
\usepackage{mathrsfs}
\usepackage{nicefrac}
\newtheorem{theorem}{Theorem}[section]
\newtheorem{proposition}[theorem]{Proposition}
\newtheorem{definition}[theorem]{Definition}
\newtheorem{lemma}[theorem]{Lemma}
\newtheorem{corollary}[theorem]{Corollary}

\newcommand{\refp}[1]{(\ref{#1})}

\newcommand{\bs}[1]{\boldsymbol{#1}}
\newcommand{\ud}[1]{\, \mathrm{d}#1}
\newcommand{\IND}[1] {{ \mathds{1}_{ #1 }} }
\newcommand{\la}{\langle}
\newcommand{\ra}{\rangle}

\newcommand{\R}{\mathbb{R}}
\renewcommand{\P}{ \mathbb P }

\newcommand{\spec}{\text{Spec}}
\renewcommand{\Re}{\text{Re}}
\newcommand{\Dom}{\text{Dom}}

\renewcommand{\L}{\mathcal{L}}

\newcommand{\A}{\mathcal{A}}
\newcommand{\K}{\mathcal{K}}
\newcommand{\F}{\mathcal{F}}
\newcommand{\C}{\mathcal{C}}
\newcommand{\Pe}{\mathsf{P}}
\newcommand{\Q}{\mathsf{Q}}

\newcommand{\rcrit}{r_{\text{crit}}}

\newcommand{\bphi}{\boldsymbol{\phi}}
\newcommand{\e}{\boldsymbol{e}}

\newcommand{\vp}{\varphi}

\renewcommand{\a}{\alpha}
\renewcommand{\b}{\beta}
\newcommand{\nuu}{\nu_1}

\renewcommand{\geq}{\geqslant}
\renewcommand{\leq}{\leqslant}
\renewcommand{\d}{\partial}
\newcommand{\derivAD}[3][]{\frac{\ud^{#1} \hspace{-0.3mm} #2}{\ud_{AD}{#3}^{#1}}}
\newcommand{\deriv}[3][]{\frac{\ud^{#1} \hspace{-0.3mm} #2}{\ud{#3}^{#1}}}
\newcommand{\pderiv}[3][]{\frac{\d^{#1} \hspace{-0.1mm} #2}{\d{#3}^{#1}}}
\renewcommand{\div}[1][]{\nabla_{\!\! #1}\cdot \!}
\newcommand{\grad}[1][]{\nabla_{\! #1} }

\begin{document}

\title{Population persistence under advection-diffusion in river networks}


\author{Jorge M Ramirez \footnote{Universidad Nacional de Colombia, Sede Medellin. email: jmramirezo@unal.edu.co}}

\maketitle

\begin{abstract}
An integro-differential equation on a tree graph is used to model the evolution and spatial distribution of a population of organisms in a river network. Individual organisms become mobile at a constant rate, and disperse according to an advection-diffusion process with coefficients that are constant on the edges of the graph. Appropriate boundary conditions are imposed at the outlet and upstream nodes of the river network. The local rates of population growth/decay and that by which the organisms become mobile, are assumed constant in time and space. Imminent extinction of the population is understood as the situation whereby the zero solution to the integro-differential equation is stable. Lower and upper bounds for the eigenvalues of the dispersion operator, and related Sturm-Liouville problems are found, and therefore sufficient conditions for imminent extinction are given in terms of the physical variables of the problem.
\end{abstract}

\section{Introduction}

The problem of persistence of a population of organisms in an environment with predominantly unidirectional flow gives rise to the so-called ``drift paradox'' \citep{Muller:1982fk}. For riverine habitats, the problem is to determine physical, biological or dynamical mechanisms by which a species whose organisms spend time drifting in the water column, avoids being driven to extinction. Recently, a body of quantitative work has been used to address this problem,  see \citep{Lutscher:2005p1894} and references therein. In particular in \cite{Lutscher:2005p1894} the authors used the following integro-differential equation to model spatio-temporal dynamics of the population, at low-density values, on a river stretch of length $l$: 
\begin{equation}\label{IDE1d}
\pderiv{u}{t}(x,t) = r \, u(x,t) - \mu \, u(x,t) + \mu \int_0^l \K(y,x) u(y,t) \ud y.
\end{equation}
Here, $u(x,t)$ is the number of individuals per unit length at point $x$ and time $t$, $r>0$ is a net population growth rate at small densities, $\mu>0$ is the rate at which individuals become mobile, and $\K$ is a ``dispersion kernel'', namely $\K(y,x)$ is the probability that a mobile individual disperses from $y$ to $x$.

Of main interest is to find conditions under which the trivial solution $u \equiv 0$ is a stable state of equation \refp{IDE1d}. Under those conditions we say that the population faces ``imminent extinction'', namely it cannot endure low population density values. The complementary situation, where stability of the zero solution to \refp{IDE1d} does not hold, is referred to as ``persistence'' \citep{Jrgensen:2004p6218}.

The integro-differential Equation \refp{IDE1d} is of the Barbashin type,  and a necessary condition for the stability of the trivial solution is that its largest Lyapunov exponent be negative \citep{Appell:2000p6353}. Namely,
\begin{equation}\label{stability}
r - \mu + \mu \sup\{ \Re(\omega);\, \omega \in \spec(\K) \} < 0,
\end{equation}
where the spectrum $\spec(\K)$ of the integral operator $\K[f] = \int_0^l \K(y,x) f(y) \ud y$  denotes the collection of all its eigenvalues, that is, complex numbers $\omega$ such that $\K[f] = \omega f$ for some $f$ in the appropriate function space.

The analytical work in \cite{Lutscher:2005p1894} yielded the following results: (i) If the dispersal kernel does not depend on the size of the environment $l$, the largest eigenvalue $\omega_\K$ of $\K$ is an increasing function of $l$, and therefore, there exists a ``critical domain size'' $\bar l$ under which inequality \refp{stability} does not hold, and the population faces imminent extinction. (ii) If the motion of mobile individuals is assumed to follow an advection-diffusion process with drift $v$ and diffusion coefficient $D>0$, then upstream dispersal is likely enough to ensure the  existence of persistence scenarios, even for high values of $v$. These results rely on the assumption that the environment is a single river stretch and ignore boundaries of the environment; in particular the conditions in (ii) are found for an advection-diffusion kernel on the whole real line truncated to $[0,l]$.

In this paper we extend the analysis of \cite{Lutscher:2005p1894} in two directions. First the finiteness of the environment is taken into account and particular boundary conditions are considered. Secondly, equation \refp{IDE1d} is considered  in a directed tree graph made of several connected  segments with different physical properties, and hence arrive at conditions for persistence of populations undergoing advection-diffusion in a river network. The goal is to lay the foundation for a deeper understanding of the role that some of the properties of real river networks, e.g. heterogeneity, connectedness, scaling, etcetera,  play on the existence of persistence conditions.

The rest of the introduction is devoted to revisiting the conditions of stability on the one-dimensional case, and summarizing the results in the case of a network. In section \ref{SectionAD} the equations and operators for advection-diffusion on river networks are derived. Section  \refp{SectAnal} deals with the relationship between $\K$ and a related Sturm-Liouville operator, along with some important consequences. Lastly, sections \refp{SectVariational} and \refp{SectOscillation} are devoted to proving the stability conditions for populations on networks.

\subsection{The one-dimensional case} \label{Section1d}
As a motivation, and to fix ideas, assume the population density evolves on a river stretch $[0,l]$ as in \refp{IDE1d}. Consider also the following particular dispersion mechanism. Let $P(y,x,t)$ be the transition probabilities of a diffusion process with constant diffusion $D$, drift $v$; with absorbing boundary condition at $x=0$, and reflecting boundary conditions at $x=l$. Then $P$ satisfies the backwards equation,
\begin{equation}\label{p}
\pderiv{P}{t}  = D \pderiv[2]{P}{y} - v \pderiv{P}{y}, \quad P(0,t) = \pderiv{P}{x}(l,t) = 0.
\end{equation}
Suppose now that mobile organisms follow the diffusion process described in \refp{p} for an exponentially distributed random time of mean $\frac{1}{\sigma}$ units of the time scale $t$ defined by \refp{IDE1d}. Then, the dispersion kernel is
\begin{equation}\label{K1d}
\K(y,x) = \int_0^\infty \sigma e^{-\sigma t} P(y,x,t) \ud t.
\end{equation}

Estimates on the eigenvalues of $\K$ can be given in terms of the following two non-dimensional quantities:
\begin{equation}\label{nondim}
\Pe : =\frac{v l}{D}, \quad \Q:= \frac{v}{\sigma l}.
\end{equation}
The first one being the classical Pecl\`{e}t number, measuring the relative importance of advection and diffusion in the dispersion process. The variable $\Q$ compares the mean velocity of the channel with the velocity needed to transverse the whole channel during the mean time individuals are mobile. 

The results in the one-dimensional case follow from standard theory. However, details are included here to illustrate some of the ideas that will be later used in the case or river networks. 

\begin{theorem}\label{ThmStability1d}
The largest eigenvalue is given by $\omega_\K = \frac{1}{\nuu}$ where $\nuu$ satisfies
\begin{equation}\label{upperbound1d}
 1+\frac{1}{4}\Q\Pe< \nuu < 1+\frac{1}{4}\Q\Pe + \frac{\pi^2}{4}\frac{\Q}{\Pe}.
 \end{equation}
\end{theorem} 

\begin{proof}
Any eigenvalue for the operator $\K$ in \refp{K1d} is given by $\omega= \frac{1}{\nu}$ where $\nu$ is a $q$-eigenvalue of the associated Sturm-Liouville operator
\begin{equation}
\L[f] = (-pf')'+qf, \quad p(x) := e^{-\frac{v}{D}x}, \; q(x) := \frac{\sigma}{D} p(x).
\end{equation}
Namely, $\nu$ satisfies $\L[u] = \nu q u$ for some $u \in \C^2([0,l])$ such that $u(0) = u'(l) = 0$. In particular, $u$ is has the form $u(x;\nu) = B e^{\a x} + C e^{\b x}$ for $B,C \in \R$, and $\a,\b$ given by 
$$\alpha := \frac{v+\sqrt{v^2-4D \sigma(\nu-1)}}{2D}, \quad 
\beta := \frac{v-\sqrt{v^2-4D \sigma(\nu-1)}}{2D}.$$
Evaluation of the boundary conditions yields that both $\a$ and $\b$ must have non-zero imaginary part, and so the estimate on the left hand side of \refp{upperbound1d} holds. Moreover, $\nu_1$ can be obtained as the smallest solution to
\begin{equation}\label{Eqnu1d}
\tan(lb(\nu)) + \frac{lb(\nu)}{\Pe} = 0,\quad \nu = \frac{(lb(\nu))^2+\Pe^2/4}{\Pe/\Q} + 1
\end{equation}
where $b(\nu)= \frac{1}{2D}\sqrt{|v^2-4D \sigma(\nu-1)|}$. The first of the equations on \refp{Eqnu1d} defines $b$ as a function of $l$, and one can differentiate implicitly to obtain $\deriv{b}{l} < 0$. Also $\deriv{\nu}{b} = \frac{2 b D}{\sigma} > 0$, so it follows that $\deriv{\nu}{l} < 0$. The upper bound in \refp{upperbound1d} follows simply from noting that since $\frac{lb}{\Pe} >0$, the smallest positive solution to $\tan(lb) + \frac{lb}{\Pe} =0$ must satisfy $lb \in (\tfrac{\pi}{2},\pi)$.
\end{proof}

The following corollaries contain the extinction and critical domain conditions that follow from theorem \refp{ThmStability1d}.

\begin{corollary}\label{CorStab1D}
\begin{enumerate}
\item If $r>\mu$ the population will persist. 
\item Given $v,D,\sigma$ and $l$, there is a critical population growth rate $\rcrit:=1-\frac{1}{\nu_1}$, such that if $r<\rcrit$, the population will face imminent extinction. Moreover,
\begin{equation}\label{condition1d}
\frac{\Pe \Q}{4 + \Pe \Q} < \frac{\rcrit}{\mu} < 1-\frac{4 \Q}{\Pe \left(\Q^2+\pi ^2\right)+4 \Q}.
\end{equation}
\item Given $0<r<\mu$ and fixing all other parameters, there is a critical length $l_\text{crit}$ such that if $l<l_\text{crit}$, the population will face imminent extinction.
\end{enumerate}
\end{corollary}

Figures \refp{FigRcrit} and \refp{FigPcrit} show the dependence of the critical values in terms of non-dimensional quantities: $\Pe$, $\Q$, $\Pe \Q = \frac{v^2}{D\sigma}$ and $\frac{r}{\mu}$.

\begin{figure}
  \includegraphics[scale=0.3]{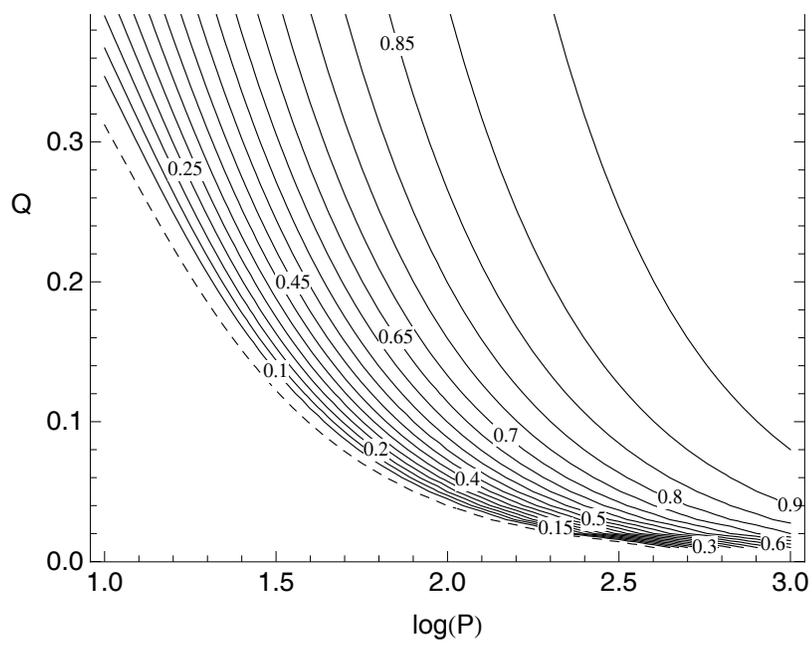}
\caption{Values of $\rcrit$ as a function of $\Pe$ and $\Q$. Note that the horizontal scale is logarithmic. The dashed line shows the boundary of positive critical reproductive rate.}
\label{FigRcrit}       
\end{figure}

\begin{figure}
  \includegraphics[scale=0.4]{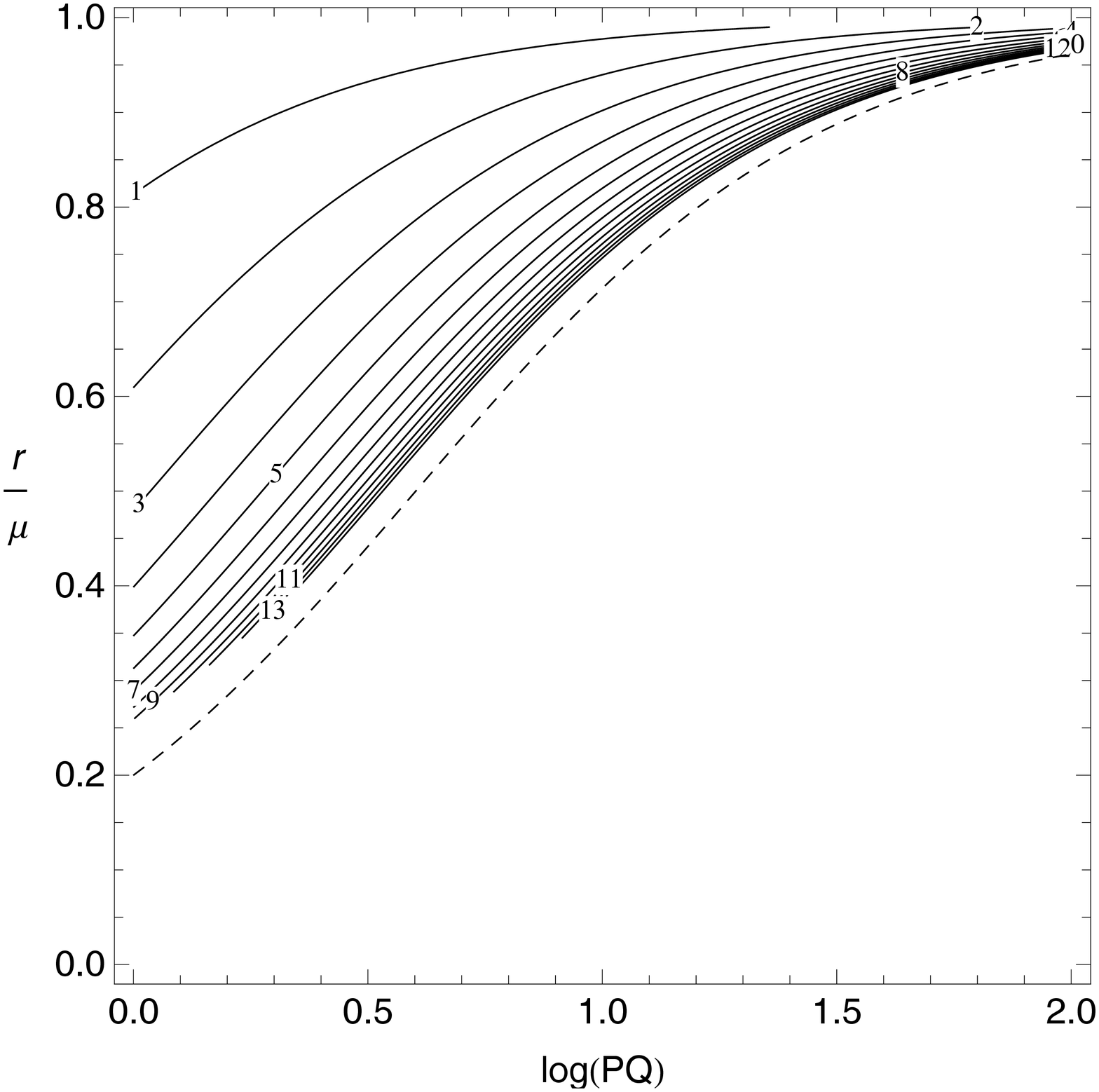}
\caption{Critical P\`{e}clet number $\Pe_\text{crit}$. Given $\frac{r}{\mu}$ and fixed $v,D$ and $\sigma$, the value $\Pe_\text{crit} = \frac{vl_\text{crit}}{D}$ gives the length of channel below which extinction is imminent. Note the horizontal logarithmic scale. The dashed line corresponds to the lower bound on the left hand side of \refp{condition1d}}
\label{FigPcrit}       
\end{figure}

\subsection{Dispersion on networks.}
By a network $\Gamma$ we understand a directed, finite, binary, geometric graph embedded in $\R^2$. We assume that each edge $e \in \Gamma$ allows a sufficiently smooth parametrization, contains no self-intersections, and is finite, therefore can be considered as the interval $e = [0,l_e]$. A point on $\Gamma$ is then denoted as the pair $(e,x)$ with $0 \leq x \leq l_e$. The ``root edge'' is denoted as $r$, and each edge has either zero or two children edges connected at the point $(e,l_e)$. At each endpoint of an edge is located a node of $\Gamma$. The set of nodes is $N(\Gamma)$ and boldface is used to denote individual nodes. The ``upstream node'' of edge $e$ is $\e = (e,l_e)$; the ``root node'' of $\Gamma$ is $\bs{\phi} = (r,0)$.

 The children edges of $r$ are $\la 0\ra$ and $\la 1 \ra$. Inductively, the children edges of some edge $e = \la i \ra$ are $\la i0 \ra$ and $\la i1 \ra$, also denoted as $e0$ and $e1$. The upstream node of $e$ has three possible representations
$$\bs e = (e,l_e) = (e0,0) = (e1,0).$$ 
We say that $e$ is a ``leaf edge'' if it has no children edges. Upstream nodes of leaf edges form the upstream boundary of $\Gamma$, denoted as $U(\Gamma)$. The ``boundary'' of $\Gamma$ is defined as $\d \Gamma = U(\Gamma) \cup \{\bs{\phi}\}$. Internal nodes are $I(\Gamma) = N(\Gamma)\smallsetminus \d\Gamma$.

Values of a function $f:\Gamma \to \R$ are denoted by $f_e(x) = f(e,x)$. That is, $f_e$ is the restriction of $f$ to the edge $e$. Conversely, given any collection of functions $\{f_e:[0,l_e]\to \R; \, e \in \Gamma\}$ with the property $f_e(l_e) = f_{e0}(0) = f_{e1}(0)$ for all $e \in I(\Gamma)$, one can construct the function $f = \sum_{e \in \Gamma} f_e \IND{e}$ on $\Gamma$. The derivative of a function $f:\Gamma \to \R$ at a point $(e,x)$ in the interior of $e$ is understood in the usual sense. For edge endpoints $(e,0)$ and $(e,l_e)$ the derivative $f'$ is understood as the right and left derivatives respectively. At an internal node $\e$, the derivative of a function is not well defined as it can take three values. For nodes $\e \in \d \Gamma$, $f'(\e)$ denotes the appropriate one-sided derivative. 

The space of functions that are continuous inside each edge of $\Gamma$ and at all nodes, is denoted by $\C(\bar \Gamma)$. The set $\C(\Gamma)$ contains, in contrast, those functions that are continuous for $(0,l_e)$ for each edge, but such that their values at nodes may not correspond to the one-sided limit taken along an incoming edge. Similarly, $\C^n(\Gamma)$, $n=1,2,\dots,\infty$ denotes the space of functions that are $n$ times continuously differentiable in the interior of all edges.  

A tree graph $\Gamma$ will serve as a model for a river network with outlet located at $\bs{\phi}$, and each edge is the one-dimensional representation of a stream. Given an edge $e$, the following strictly positive variables are defined: water velocity $v_e$, cross-sectional area $A_e$, and diffusion coefficient $D_e$. We are then assuming that these correspond to representative spatio-temporal mean values of the physical variables on each stream in the network. The following conservation of water flux plays an important role
\begin{equation}\label{consWater}
A_{e0} v_{e0} + A_{e1} v_{e1} = A_e v_e, \quad \e \in I(\Gamma).
\end{equation}
the functions $A$, $v$ and $D$ are defined in the interior of each edge as (for example) $D(x) = D_e$ for $x \in e$; at internal nodes they are defined as (say) equal to zero, and at boundary nodes they are defined as the value of the corresponding variable in the node's edge.

As in \refp{Section1d}, assume that mobile individuals follow random trajectories according to a diffusion process on $\Gamma$ with drift $v$ and diffusivity $D$, reflecting conditions on the upstream boundary nodes of the network and absorbing at the outlet node \citep[see][]{Freidlin:1993p4685}. While in the interior of edge $e$, diffusing particles will follow an usual diffusion process with drift and diffusion coefficients equal to $v_e$ and $D_e$ respectively. If $P(y,x,t)$, $x,y \in \Gamma$, $t>0$ is the family of transition probability densities for the aforementioned diffusion process, then at any internal node $\e$, and as a function of the backwards variable, $P(\cdot,x,t)$ must be continuous and the total diffusive flux must be equal to zero. 

More precisely, for any function $f:\Gamma \to \R$, and $\e \in I(\Gamma)$, define the total flux at $\e$ by
\begin{equation}
\derivAD{f}{}(\e) = A_eD_e f_e'(l_e) - A_{e0}D_{e0} f_{e0}'(0) - A_{e1}D_{e1} f_{e1}'(0),
\end{equation}
and consider the infinitesimal operator 
\begin{equation}\label{A}
\A[f] = D_e f_e'' - v_e f_e' \text{ on edge } e, \quad e \in \Gamma,
\end{equation}
with domain
\begin{equation}\label{DomA}
\begin{split}
\Dom(\A) = \{ f \in \C(\bar \Gamma) \cap \C^2(\Gamma):& \; f(\bphi) = 0, \; f'(\e)=0 \text{ for all } \e \in U(\Gamma), \\
& \; \derivAD{f}{}(\e) = 0, \text{ for all } \e \in I(\Gamma)\}.
\end{split}
\end{equation}
For any $x \in \Gamma$, define $P(\cdot,x,t)$ as the solution to the ``backwards'' equation
\begin{equation}\label{P}
\pderiv{P}{t}(\cdot,x,t) = \A[P(\cdot,x,t)].
\end{equation}

As in the one-dimensional case described in Section \refp{Section1d}, let $\sigma >0$  and define the dispersion kernel by
\begin{equation}\label{Kxy}
\K(y,x) = \int_0^\infty \sigma e^{-\sigma t} P(y,x,t) \ud t, \quad x,y \in \Gamma.
\end{equation}
Performing a change of time scales $t_{new} = \mu t$ on the analog of equation \refp{IDE1d}, writing $\bar r:= r/\mu$ and retaining the symbol $t$ for time, gives the integro-differential equation governing a population density $u$,
\begin{equation}\label{IDE}
\pderiv{u}{t}(x,t) = (\bar r -1) u(x,t) + \int_\Gamma \K(y,x) u(y,t) \ud y.
\end{equation}
Equation \refp{IDE} is to be understood as an evolution equation in the Banach space $\C(\bar\Gamma)$, and thus, one can characterize the stability of the trivial solution $u\equiv 0$ via the spectrum of the operator 
\begin{equation}\label{K}
\K[f] := \int_\Gamma \K(y,x) f(y) \ud y.
\end{equation}
More precisely \citep{Appell:2000p6353},

\begin{theorem}
Let $\omega_\K$ be the largest eigenvalue of the operator $\K$ given in \refp{K}. The inequality 
$$\bar r - 1 + \omega_\K <0$$
is sufficient for the exponential stability of the trivial solution $u \equiv 0$ to \refp{IDE}.
\end{theorem}
 
\subsection{A related Sturm-Liouville problem}
The dispersion kernel $\K$ is closely related to a Sturm-Liouville problem on $\Gamma$. Consider the functions
\begin{equation}\label{pq}
p(x) := \exp \left\{ -\int_{\bphi}^x \frac{v(y)}{D(y)} \ud y \right\}, \quad q(x) = \frac{\sigma p(x)}{D(x)},
\end{equation}
where the integral on the definition of $p$ is taken along the unique path connecting the root $\bphi$ and the point $x \in \Gamma$. 

Define the following Sturm-Liouville operator,
\begin{equation}\label{L}
\L[f] = -(p_e f'_e)' + q_e f_e \text{ on edge } e, \quad e \in \Gamma,
\end{equation}
with domain 
\begin{equation}\label{DomL}
\Dom(\L) = \left\{ f \in \C(\bar \Gamma) \cap \C^2(\Gamma); \derivAD{f}{}(\e) = 0 \text{ for all } \e \in I(\Gamma)\right\}.
\end{equation}
The specification of ``hydrological'' boundary conditions is made considering functions on the following class:
\begin{equation}\label{BH}
B_H = B_H(\Gamma) = \{f:\Gamma \to \R; \; f(\bphi) = 0, \; f'(\e) = 0 \text{ for all } \e \in U(\Gamma) \}
\end{equation}
The interest here is mostly on functions that belong to $\Dom(\L) \cap B_H = \Dom(\A)$.

The next theorem describes how $\K$ and $\L$ are related. 

\begin{theorem}\label{ThmLK}
Let $\L$ and $\K$ be defined as in \refp{L} and \refp{K} respectively.
Then $\omega$ is an eigenvalue of $\K$, if and only if $\nu = \frac{1}{\omega}$ is a $q$-eigenvalue of $\L$ restricted to $B_H$, namely, there exists a solution $u$ to the following problem
\begin{equation}\label{spectral}
u\in \Dom(\L)\cap B_H, \quad \L[u] = \nu q u.
\end{equation}

\end{theorem}

General theory about the Sturm-Liouville problem on graphs, including the existence and representation of  solutions, Green's functions, and information about eigenvalues and eigenfunctions, can be found in recent literature. In particular, combining general results by \cite{Pokornyi:2004p3263} and \cite{VonBelow:1988p3428} one gets 
\begin{theorem}
All eigenvalues of problem \refp{spectral} are real, bounded below by one, and form a unbounded discrete set.
\end{theorem}
And its immediate corollary,
\begin{corollary}
\begin{enumerate}
\item If $r>\mu$ then the population persists.
\item Let $\nuu(\Gamma)$ be the smallest eigenvalue of problem \refp{spectral}. Then 
\begin{equation}
\rcrit := \mu \left(1- \frac{1}{\nuu(\Gamma)} \right)
\end{equation}
is the ``critical reproduction rate''. For $r < \rcrit$ the population will face imminent extinction. 
\end{enumerate}
\end{corollary}

\subsection{Conditions for imminent extinction}

Calculation of the smallest eigenvalue $\nuu(\Gamma)$ to problem \refp{spectral} reduces to finding the roots of a determinant, and can be performed for small networks. In general, a more useful approach is to compute bounds for the smallest eigenvalue that will lead to easily testable sufficient conditions for imminent extinction or persistence.

Variational techniques for eigenvalue bracketing have been long used for Sturm-Liouville problems (see for example \cite{Weinberger:1974uq}), and in particular for problems on networks by \cite{Currie:2005p5932}. Here we extend and adapt the latter work to our particular form of operator $\L$.

Our first result relates the eigenvalues of problem \refp{spectral} to those of a similar problem on sub-networks of $\Gamma$. First some notation. For any edge $e \in \Gamma$, denote by $\Gamma_e$ the tree that has $e$ as its root edge. Namely $\Gamma = \Gamma_r$, and if $e$ is a ``leaf edge'' then $\Gamma_e = e$. 

\begin{theorem}\label{ThmSubtree}
Let $e \in \Gamma$, and define functions $p^{(e)},q^{(e)}$ and the operator $\L^{(e)}$ by \refp{pq}, \refp{L} and \refp{DomL} for the subtree $\Gamma_e$. Let $\nuu(\Gamma_e)$ be the smallest eigenvalue of the problem $\L^{(e)}[u] = \nu q^{(e)} u$, $u \in \Dom(\L^{(e)}) \cap B_H(\Gamma_e)$. Then $\nuu(\Gamma) \leq \nuu(\Gamma_e)$.
\end{theorem}

In particular, for an upstream sub-network $\Gamma_e$ of $\Gamma$, the respective critical reproductive rates satisfy $\rcrit(\Gamma) \leq \rcrit(\Gamma_e)$. This implies that any restriction of the population to an upstream sub-habitat can only icrease the minimum reproductive rate needed for persistence. Also, the following estimate follows from theorem \refp{ThmStability1d},
\begin{equation}\label{upbound}
\nuu(\Gamma) < \min \left\{ 1+\frac{1}{4}\Q_e\Pe_e + \frac{\pi^2}{4}\frac{\Q_e}{\Pe_e},\; e \text{ is a leaf edge}\right\},  
\end{equation}
where $\Pe_e:= \frac{v_e l_e}{D_e}$ and $\Q_e:=\frac{v_e}{\sigma_e l_e}$ are the non-dimensional quantities \refp{nondim} defined now on each edge.

Lower bounds for $\nu_1(\Gamma)$ can be found relating problem \refp{spectral} to the similar spectral problem with Dirichlet (absorbing) boundary conditions on all boundary nodes of $\Gamma$. Define the class of functions
\begin{equation}\label{DomD}
B_D = B_D(\Gamma) = \{f:\Gamma \to \R; \; f(\e) = 0 \text{ for all } \e \in \d\Gamma \}
\end{equation}
and consider the spectral problem
\begin{equation}\label{spectralDir}
u \in \Dom(\L)\cap B_D, \quad \L[u] = \eta q u.
\end{equation}
Much more is known about the properties of the solutions to the spectral problem \refp{spectralDir} than those of problem \refp{spectral}. In particular the set of eigenvalues of \refp{spectralDir} is real, discrete, unbounded, and the smallest eigenvalue $\eta_1(\Gamma)$ has an associated eigenfunction that is strictly positive in the interior of $\Gamma$ \citep{Pokornyi:2004p3263,VonBelow:1988p3428}.

The spectral problems \refp{spectral} and \refp{spectralDir} with hydrological and Dirichlet boundary conditions can be related using variational techniques, and in particular one gets 
$$\nuu(\Gamma) < \eta_1(\Gamma),$$
which doesn't prove to be very useful, since a calculation of $\eta_1(\Gamma)$ is as difficult as that of $\nuu(\Gamma)$. However, the next key lemma offers a different relation between the two problems.

\begin{lemma}\label{LemmaExtend}
There exists a network $\tilde\Gamma$ containing $\Gamma$, and differing from $\Gamma$ only in the length of leaf edges, such that $\eta_1(\tilde \Gamma) \leq \nuu(\Gamma)$.
\end{lemma}

As in \cite{Pokornyi:2004p3263}, lower bounds for $\eta_1(\tilde \Gamma)$ are possible via considerations about the oscillation of eigenfunctions for problem \refp{spectralDir} on $\tilde \Gamma$. Our main result follows from such techniques.

\begin{theorem}\label{ThmLowBound}
The smallest eigenvalue $\nuu(\Gamma)$ for problem \refp{spectral} satisfies
\begin{equation}
\nu_1(\Gamma) >  \min_{e\in\Gamma} \frac{v_e^2}{4D_e \sigma} + 1, \quad \text{ for all } e\in \Gamma.
\end{equation}
And in particular, in view of \refp{upbound},
\begin{equation}
\min_{e\in \Gamma} \frac{\Pe_e\Q_e}{4+ \Pe_e \Q_e}< \frac{\rcrit}{\mu} < 
\min_{e \text{ is leaf}} \; 1-\frac{4 \Q_e}{\Pe_e \left(\Q_e^2+\pi ^2\right)+4 \Q_e}
\end{equation}

\end{theorem}

\section{Advection-diffusion dynamics on graphs}\label{SectionAD} 

For ease of notation, assume that the habitat consists of only three channels: two smaller channels converge to form another one. The mathematical model for this small river network is $\Gamma$, a graph  containg three edges: $r$, $\la 0 \ra$ and $\la 1 \ra$. And three nodes: the point where the streams converge $\bs{r}$, two upstream boundary nodes $\bs{\la 0 \ra}$, $\bs{\la 1 \ra}$, and the outlet $\bphi$.  Each channel is in reality a three-dimensional domain with coordinates $(x,y,z)$ where the longitudinal variable takes values between $[0,l_e]$, $x=0$ corresponding to the downstream end of the channel, and $l_e$ being the length of the channel. 

In full generality, organisms dispersing in a channel $e$ can be quantified by a volumetric concentration $c_e(x,y,z,t)$ given by the number of mobile individuals in a representative unit volume of water. Let $\mathbf{v}_e(x,y,z,t)$ denote the water velocity field. Assuming that the mechanics of mobile individuals follow a linear advection-diffusion process, then there exists a positive-definite diffusion tensor $\mathbf{D}_e(x,y,z,t)$ such that
\begin{equation}\label{3du}
\pderiv{c_e}{t} = \div(\mathbf{D}_e \grad c_e) - \div(\mathbf{v}_e c_e). 
\end{equation}

Several simplifications are in order. First, assume that $\mathbf{v}_e = (-v_e,0,0)$ where $v_e>0$ is constant, and the diffusivity tensor is the constant matrix $\mathbf{D}_e = D_e \mathbf{I}_{3 \times 3}$. The transversal variables $y,z$ vary over the cross-sectional area of the channel which we will assume constant $A_e(x) = A_e$. The concentration of individuals per unit length in edge $e$,
\begin{equation}\label{uc}
u(x,t)  := \int_{A_e} c_e(x,y,z,t) \ud y \ud z
\end{equation}
satisfies the usual one-dimensional advection-diffusion equation 
\begin{equation}\label{PDEu}
\pderiv{u_e}{t} = D_e \pderiv[2]{u_e}{x} + v_e \pderiv{u_e}{x}, \quad x\in e, \; t>0.
\end{equation}

The conditions on $u$ at the intersection of the channels are now derived. In the interior of each channel, the continuity of the concentration $c$ is implied by equation \refp{3du}. The junction of the channels is a three-dimensional volume depicted by the shaded area in Figure \refp{FigJunction}. Assuming that $c$ is constant there, yields the jump condtion of $u$ at node $\bs{r}$, 
\begin{equation}\label{contu}
\frac{u_r(l_r,t)}{A_r}  = \frac{u_{\la 0\ra}(0,t)}{A_{\la 0\ra}} = \frac{u_{\la 1\ra}(0,t)}{A_{\la 1\ra}}
\end{equation}

\begin{figure}
  \includegraphics[scale=0.4]{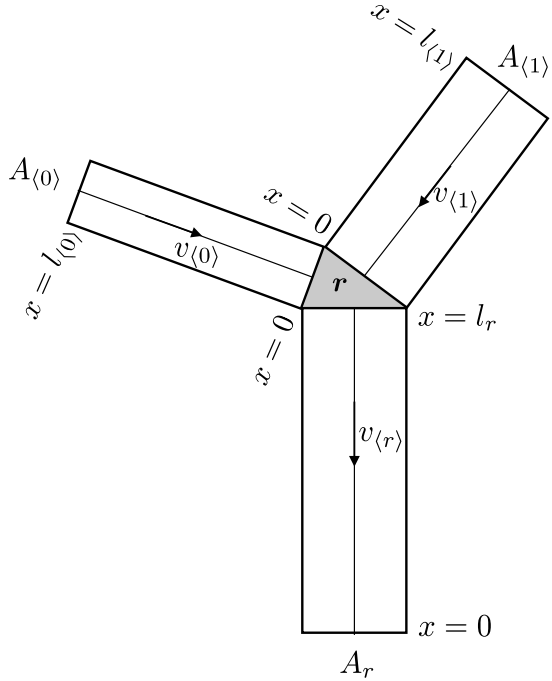}
\caption{Schematic representation of a junction of two channels. The corresponding model is a graph with three edges: $r, \la0\ra$ and $\la 1 \ra$ connected at node $\boldsymbol{r}$.}
\label{FigJunction}       
\end{figure}

Two variables are conserved at the junction: the water flow, and the flux of mobile individuals. The former condition yields
\begin{equation} \label{consWater2}
A_{\la 0 \ra} v_{\la 0 \ra} + A_{\la 1 \ra} v_{\la 1 \ra} = A_{\la r \ra} v_{\la r \ra}.
\end{equation}
Also, if storage and other local population changes are neglected, the total amount of individuals entering the junction volume must be equal to the amount leaving it. Therefore,
\begin{equation}
\sum_{e=\la0\ra,\la1\ra} A_e \left[ D_e \pderiv{c_e}{x} + v_e c_e\right|_{x=0} = A_r \left[ D_r \pderiv{c_r}{x} + v_r c_r\right|_{x=l_r},
\end{equation}
which in view of \refp{consWater2} and \refp{contu} yields
\begin{equation}\label{Fluxu}
D_{\la 0 \ra} \pderiv{u_{\la 0 \ra}}{x}(0) + D_{\la 1 \ra} \pderiv{u_{\la 1 \ra}}{x}(0) = D_{r} \pderiv{u_r}{x}(l_r).
\end{equation}

This may be viewed as an interface condition generalizing notions
of ``one-dimensional skew diffusion''  to a tree graph; see \cite{Appuhamillage:2011uq, Ramirez:2011kx}, and references therein for a discussion of surprising consequences of these and other such interface  conditions on the dispersion of particles across an interface
in one-dimension.

Now we address boundary conditions for equation \refp{PDEu} on the small river network $\Gamma$. First, it is assumed that any mobile organism that reaches the outlet $\bphi$ leaves the habitat forever. In the diffusion process language, this is referred to as an ``absorbing boundary'' and its formalized by imposing
\begin{equation}\label{absorbu}
u(\bphi,t) = u_r(0,t) =0, \quad t>0. 
\end{equation}
  
The upstream boundaries of edges $\la 0\ra$ and $\la 1 \ra$ correspond to the headwaters of the network. Its assumed that no organism crosses those points, namely the total flux is equal to zero,
\begin{equation}\label{reflectu}
A_e \left[D_e \pderiv{u_e}{x} + v_e u_e \right|_{x = l_e} \!\!\!= 0, \quad e=\la0\ra, \la 1 \ra.
\end{equation}

Consider now a general network $\Gamma$. Denote by $(\A^*,\Dom(\A^*))$ the ``forward'' operator defined by \refp{PDEu} acting on those functions in $\C^2(\Gamma)$ that satisfy \refp{contu} for all nodes in $I(\Gamma)$, \refp{absorbu} at $\bphi$, and  \refp{reflectu} for nodes in $U(\Gamma)$. The formal adjoint of $(\A^*,\Dom(\A^*))$ is the operator $(\A,\Dom(\A))$ defined in \refp{A}. Moreover, in \cite{Freidlin:2000vn} it is shown that $\A$ is the infinitesimal generator of a strongly continuous semigroup of linear operators on $\C(\bar \Gamma)$ corresponding to a conservative Markov process $X=\{X(t):t\geq 0\}$ with continuous sample paths. Namely, if $P(y,x,t)$, $x,y\in \Gamma$, $t>0$, denotes the family of transition probability densities defined as the solution of \refp{P}, then the semigroup 
\begin{equation}\label{Tt}
u(y,t) = T_t[h](y) := \int_\Gamma h(x) P(y,x,t) \ud x, \quad t>0,
\end{equation}
is the solution to 
$$\pderiv{u}{t} = \A[u], \quad u(y,0) = h(y).$$
As in the classical case \citep[see][]{Bhattacharya:1990kx}, if one defines the adjoint semigroup by
\begin{equation}\label{Tstar}
T_t^*[f] = \int_\Gamma f(y) P(y,x,t) \ud x, 
\end{equation}
then $\A^*[T_t^*[f]] = f$. Namely, for a suitable initial condition $f\geq0$ with $\int_\Gamma f \ud x =1$, $T_t^*[f]$ given in  \refp{Tstar} gives the evolution of the population densitiy. Moreover, the family of transition probability densities has the following properties: $P(\cdot,x,t) \in \Dom(\A)$, $P(y,\cdot,t) \in \Dom(\A^*)$ for all $x,y \in \Gamma$, $t>0$. As $P_y :=P(y,\cdot,\cdot)$ solves $\pderiv{P_y}{t} = \A^*[P]$, $\lim_{t\to0}P(y,\cdot,t) = \delta_y$ for all $y \in \Gamma$, then for $A\subset \Gamma$, $\int_A P(y,x,t) \ud x$ is equal the fraction of individuals that where initially at $y$ and occupy $A$ at time $t$, namely is equal to $\P(X(t) \in A \large| X(0) = y)$.     

Let $\tau$ be an exponential random variable with mean $1/\sigma$, and independent of the process $X$. Then 
\begin{equation}\label{K2}
\K(y,x) := \P(X(\tau) \in \ud x | X_0 =y) = \int_0^\infty \sigma e^{-\sigma t} P(y,x,t) \ud t 
\end{equation}
is the transition probability density of a jump process on $\Gamma$ starting at $y$.

The definition \refp{K2} for the dispersion kernel $\K$ can be viewed as a time homogenization of the ``microscopical'' dynamics described by the transition probabilities $P$. Namely, the $\K[u]$ term in equation \refp{IDE} describes the dispersion of organisms at time scales comparable to those of population growth. At this scale, the continuous motion individuals in the water column scales-up to a Markov jump process on $\Gamma$ with infinitesimal generator given by the operator on the right hand side of \refp{IDE}.

\section{Analysis of the dispersion kernel} \label{SectAnal}
Stability results for equation \refp{IDE} will follow from relating the dispersion kernel $\K$ to the Sturm-Liouville operator $\L$ defined in \refp{L} with $p$ and $q$ as in \refp{pq}. By the definition \refp{Kxy} of $\K(y,x)$ and the strong continuity of of the semigroup $\{T_t: t \geq 0\}$ it follows that for $f \in \C(\bar \Gamma)$, 
\begin{equation}
\int_\Gamma \K(\cdot,x) \, f(x) \ud x \in \Dom(\A), \quad (\sigma - \A) \left[\int_\Gamma \K(\cdot,x) \, f(x) \ud x \right ] = \sigma f.
\end{equation}
Defining the operator 
\begin{equation}\label{L2}
\L := \frac{p}{D} (\sigma - \A)
\end{equation}
one gets the familiar form of the Sturm-Luiville operator \refp{L} on each edge.

In \cite{Pokornyi:2004p3313} it is shown that for very general boundary and internal node conditions, one can find solutions to Sturm-Liouville problems in geometric graphs. In particular, if $f \in \C(\bar\Gamma)$ then,
$$\L\left[ \int_\Gamma \K(y,x) q(x) f(x) \ud x \right] = f(y).$$
Referring back to the ``hydrological conditions'' $B_H$ in \refp{BH}, one therefore has that the solution $u\in \Dom(\L)\cap B_H$ to $\L[u] = f$ can be found via the Green's function 
\begin{equation}\label{GK}
G(y,x) := q(x) \K(y,x)
\end{equation}
by making
$u(y) = \int_\Gamma G(y,x) f(x) \ud x.$
More specifically, we have the following theorem from which theorem \refp{ThmLK} follows as a corollary.

\begin{theorem}
Let $\K(x,y)$ be as in \refp{Kxy}, $\K[f](x) = \int_\Gamma \K(y,x) f(y) \ud y$, and $\L$ defined on $\Dom(\L)$ as in \refp{L}, \refp{DomL}. 
\begin{enumerate}
\item $\K(\cdot,y) \in \Dom(\A^*)$, $\K(y,\cdot) \in \Dom(\L)\cap B_H = \Dom(\A)$.
\item The kernel $\K$ satisfies the following symmetry condition
\begin{equation}\label{symK}
p(x)A(x)\K(x,y) = p(y)A(y)\K(y,x).
\end{equation}
\item Let $f \in \C(\bar\Gamma)$, then the function $u:= \K[f]$ is such that $\frac{1}{pA}u \in \Dom(\L)\cap B_H$, and $\L[\frac{1}{pA}u] = \frac{1}{AD}f$. 
\item $u \in \C(\bar \Gamma)$ satisfies $\K[u] = \omega u$ if and only if $v:= \frac{1}{pA}u$ satisfies $\L[v] = \frac{1}{\omega} v$.
\end{enumerate}
\end{theorem}
\begin{proof}
The first assertion follows from the definition of $\K(y,x)$. For (\textit{ii}) note first that $\L$ is self-adjoint with respect to $\ud AD$ in $\Dom(\L)\cap B_H$. Namely, 
\begin{equation}\label{selfadjoint}
\int_\Gamma v \, \L[u] \ud AD = \int_\Gamma u \, \L[v] \ud AD, \quad u,v \in \Dom(\L) \cap B_H.
\end{equation}
Let $f,g \in \C_0^\infty(\Gamma)$, and consider the solutions $u$ and $v$ to $\L[u] = f$ and $\L[v] = g$ respectively. Using \refp{selfadjoint} yields
$$\int_\Gamma \int_\Gamma G(x,y) g(y)f(x) A(x)D(x) \ud x \ud y = 
\int_\Gamma \int_\Gamma G(y,x) g(y)f(x) A(y)D(y) \ud x \ud y. $$
Since $f,g$ are arbitrary, the definition \refp{GK} gives the desired result. For (\textit{iii}), use \refp{selfadjoint} in $u = \K[f]$ to get,
$$\int_\Gamma \K(y,x) \frac{u(x)}{p(x)A(x)} \ud x = \frac{f(y)}{p(y)A(y)}$$
and use \refp{GK} to arrive at an expression involving $G$. The statement in (\textit{iv}) follows directly from (\textit{iii}).
\end{proof}

\section{Variational approach}\label{SectVariational}

Consider the measure $\ud AD$ on $\Gamma$ with density $\ud AD = A_e D_e \ud x$ on edge $e$, and let $L^2_{A_eD_e}(e)$ denote the space of square-integrable functions in $(0,l_e)$ with respect to the measure $A_eD_e \ud x$. Also, let $H^1(e)$ denote the Sobolev space of functions such that themselves and their first generalized derivative lie in $L^2_{AD}(e)$. The analog spaces on all of $\Gamma$ are given as direct sums:
$$L_{AD}^2(\Gamma) := \bigoplus_{e\in \Gamma} L_{A_eD_e}^2(e), \quad 
H^1(\Gamma) := \bigoplus_{e\in \Gamma} H^1(e)$$
The real inner product in $L^2_{AD}$ is denoted by $(\cdot,\cdot)_{AD}$.

 The operator $\L$ in \refp{L} can be extended to $L_{AD}^2(\Gamma)$ by replacing derivatives by generalized derivatives where necessary, an its domain is
\begin{equation}
\Dom(\L) = \left\{ f \in \C(\bar \Gamma) \cap H^1(\Gamma);\; pu' \in H^1(\Gamma),\;\derivAD{f}{}(\e) = 0, \; \e \in I(\Gamma)\right\}.
\end{equation}

We are interested in solving the following spectral problem
\begin{equation}\label{spectral2}
u \in \Dom(\L) \cap B_H, \quad \L[u] = \nu q u,
\end{equation}
and in particular, finding the smallest number $\nuu(\Gamma)$ for which a nontrivial solution $u$ of \refp{spectral2} exists.

Extending the results of \cite{Currie:2005p5932}, one can reformulate problem \refp{spectral} in variational form. Define the associated bilinear form
\begin{equation}\label{F}
\F(u,v) = \int_\Gamma p u' v' + q uv \ud AD, \quad u,v \in \Dom(\F).
\end{equation} 
The domain of $\F$ is given by
\begin{equation}\label{DomF}
\Dom(\F) = \{u \in H^1(\Gamma) \cap \C(\bar \Gamma); \; u(\bphi) = 0\}.
\end{equation}

\begin{theorem}\label{ThmVariational}
\begin{enumerate}
\item A function $u$ belongs to $\Dom(\L)\cap B_H$ and solves $\L[u] = \nu q u$ if and only if $u\in \Dom(\F)$ and $\F(u,v) = \nu (qu,v)_{AD}$ for all $v \in \Dom(\F)$.
\item The smallest $q$-eigenvalue of $\L$ is
\begin{equation}\label{nu1eqinf}
\nu_1(\Gamma) = \inf_{v \in \Dom(\F)} \frac{\F(v,v)}{(qv,v)_{AD}}.
\end{equation}
\end{enumerate}
\end{theorem} 

\begin{proof}
Let $u \in \Dom(\F)$ with  $\F(u,v) = \nu (qu,v)_{AD}$ for all $v \in \Dom(\F)$. Consider $\C_0^\infty(\Gamma) = \oplus_{e\in \Gamma} \C_0^\infty(e)$, the space of smooth functions on each edge with zero values at each node of $\Gamma$. Clearly, $\C_0^\infty(\Gamma) \subset \Dom(\F)$ and is dense in $L^2_{AD}(\Gamma)$. Consider $\F(u,\vp)$ for $\vp \in \C_0^\infty(\Gamma)$. The extension of this functional to an operator in $L^2_{AD}(\Gamma)$ is $\F(u,\cdot) = -(pu')'+qu = \nu q u$. In particular, $pu' \in H^1(\Gamma)$, and since $p$ is smooth, $u'$ can be extended continuously to nodes in $U(\Gamma)$. For arbitrary $v\in \Dom(\F)$, integration by parts yields
$$\F(u,v) = (-(pu')'+qu,v)_{AD} + \sum_{e\in \Gamma} p_e v_e u_e' A_eD_e \Big|_0^{l_e}.$$
Since the the term on the left hand side and the first term on the right both equal $\nu(qu,v)_{AD}$, the summation on the right must be equal to zero for all $v \in \Dom(\F)$. This yields $\derivAD{u}{}(\e) = 0$ for all $\e \in I(\Gamma)$, and $u'(\e) = 0$ for $\e \in U(\Gamma)$. The converse statement in (\textit{i}) follows by integration by parts.
Part (\textit{ii}) follows from standard arguments noting that \refp{nu1eqinf} implies the positive-definitness of the form $\F(\cdot,\cdot)-(q\cdot,\cdot)_{AD}$ (see \cite{Weinberger:1974uq}, page 38). 
   
\end{proof}

Recall the definition of a subtree $\Gamma_e \subseteq \Gamma$ with root edge $e$ given just before theorem \refp{ThmSubtree}. Let $\bphi_e$ be the root node of $\Gamma_e$ and define $p^{(e)}, q^{(e)}$ as in \refp{pq} where integration is taken with $\bphi_e$ instead $\bphi$ as lower limit. The operator $\L^{(e)}$ is then defined via \refp{L} and \refp{DomL} on the graph $\Gamma_e$.

\begin{proof}\textit{of theorem \refp{ThmSubtree}}. Let $\nuu(\Gamma_e)$ and $u^{(e)}\in \Dom(\L^{(e)}) \cap B_H(\Gamma_e)$ be the smallest eigenvalue, and a corresponding eigenfunction of problem \refp{spectral} on $\Gamma_e$. Denote by $\F^{(e)}$ the analogous bilinear form of \refp{F} and \refp{DomF} on the subtree $\Gamma_e$. By theorem \refp{ThmVariational},
$$\F^{(e)}(u^{(e)},u^{(e)}) = \nuu(\Gamma_e) (q^{(e)}u^{(e)},u^{(e)})_{AD(\Gamma_e)},$$
where $(\cdot,\cdot)_{AD(\Gamma_e)}$ is the inner product on $L^2_{AD}(\Gamma_e)$. 
Both sides of last equation can be multiplied by the constant $p(\bphi_e)$, and the function $u^{(e)}$ extended to all $\Gamma$ via $u:= u_e \IND{\Gamma_e} \in \Dom(\F)$, yielding $\F(u,u) = \nuu(\Gamma_e) (qu,u)_{AD}$, therefore $\nuu(\Gamma_e) \geq \nuu(\Gamma)$.
\end{proof}

Recall the spectral problem with Dirichlet boundary conditions introduced in \refp{spectralDir}. As a last application of the variational formulation, we now give an eigenvalue bracketing result similar to that of \cite{Currie:2005p5932}.

\begin{proposition}\label{PropInf}
Let $\eta_1(\Gamma)$ be the smallest eigenvalue of problem \refp{spectralDir}. Then $\nuu(\Gamma) < \eta_1(\Gamma)$.
\end{proposition}
\begin{proof}
Consider the following bilinear form $\F_D(u,v) = \F(u,v)$ for $u,v \in \Dom(\F_D) = \Dom(\F) \cap B_D$. An argument similar to that of theorem \refp{ThmVariational} characterizes $\eta_1(\Gamma)$ in terms of $\F_D$. In particular,
$$\eta_1(\Gamma) = \inf_{v \in \Dom(\F_D)} \frac{\F_D(v,v)}{(qv,v)_{AD}}
\leq \inf_{v \in \Dom(\F)} \frac{\F_(v,v)}{(qv,v)_{AD}} = \nuu(\Gamma)$$
since $\Dom(\F_D) \subseteq \Dom(\F)$. Clearly, equality cannot hold.
\end{proof}

\section{Eigenfunctions and oscillation}\label{SectOscillation}
We now compute the particular form of eigenfunctions to problem \refp{spectral}, relate them to eigenfunctions of problem \refp{spectralDir}, and use results in oscillation theory to show the lower bound estimate in theorem \refp{ThmLowBound}.

Let $m$ denote the number of edges of $\Gamma$. There are a total of $m+1$ nodes, distributed as $\#I(\Gamma) = \frac{m-1}{2}$ internal nodes, $\#U(\Gamma) = \frac{m+1}{2}$ upstream boundary nodes, and one root node. The continuity, flux-matching and boundary conditions in $\Dom(\L)\cap B_H$ correspond to a total of
$$3 \#I(\Gamma) + \#U(\Gamma) + 1 = 2m$$ 
linear functionals, say $\{\psi_i; \, i=1,\dots,2m\}$. Namely, a function $f$ belongs to $\Dom(\L) \cap B_H$ if and only if $\psi_i[f] = 0$ for all $i$. Also, if $e$ is the $i$-th edge, every solution to $(-pu_e')'+qu_e = \nu q u_e$ on $e$ and extended as zero to all of $\Gamma$, is a linear combination of
\begin{equation}
h_{2i-1}(x;\nu) := e^{\alpha_e x} \IND{e}(x), \quad h_{2i}(x;\nu):= e^{\beta_e x} \IND{e}(x), \quad i=1,2,\dots,m
\end{equation}
where 
\begin{equation}\label{alphabeta}
\alpha_e := \frac{v_e+\sqrt{v_e^2-4D_e \sigma(\nu-1)}}{2D_e}, \quad 
\beta_e := \frac{v_e-\sqrt{v_e^2-4D_e \sigma(\nu-1)}}{2D_e}.
\end{equation}
Consider the matrix $\Delta$ with entries
\begin{equation}
\Delta_{i,j}(\nu) = \psi_i(h_j(\cdot,\nu)), \quad i,j=1,\dots,2m.
\end{equation}
The function $\det(\Delta(\nu))$ is analytic in $\nu$. Moreover, $\nu$ is an eigenvalue for problem \refp{spectral} if and only if $\nu$ is a root of $\det(\Delta(\nu))=0$, and $C = (C_1,\dots,C_{2m}) \in \text{Ker}(\Delta)$ is such that 
\begin{equation}
u(x;\nu) = \sum_{j=1}^{2m} C_j h_j(x;\nu) 
\end{equation}
is nonzero. In this case, $u(\cdot;\nu)$ is an eigenfunction corresponding to $\nu$ (see \cite{Pokornyi:2004p3313}).

Lemma \refp{LemmaExtend} follows from analyzing the particular form of $u(x;\nu)$ at leaf edges. Fro the proof, we need some additional notation. First, denote by 
\begin{equation}
a_e := \frac{v_e}{2D_e}, \quad b_e(\nu) := \frac{\sqrt{|v_e^2-4D_e \sigma(\nu-1)|}}{2D_e}, \quad e \in \Gamma.
\end{equation}
Also, recall the proof of theorem \ref{ThmStability1d} and define
\begin{equation}
\nu^*(e) := 1+\frac{v_e}{4D_e\sigma}= 1+\frac{1}{4}\Q_e\Pe_e.
\end{equation}
Then $b_e(\nu^*(e)) = 0$ and $\alpha_e = a_e + b_e$, $\beta_e = a_e - b_e$ for $\nu \leq \nu^*(e)$; $\alpha_e = a_e + i b_e$, $\beta_e = a_e - i b_e$ for $\nu > \nu^*(e)$. 
 
\begin{proof}\textit{of Lemma \refp{LemmaExtend}.}
Let $e$ be a leaf edge. A solution $u$ to problem \refp{spectral} is given on $e$ by $u_e(x;\nu) = B e^{\alpha_e x} + C e^{\beta_e x}$ for some constants $B,C$. Using $u_e'(l_e)=0$ yields $C = -B \frac{\alpha_e}{\beta_e} e^{l_e(\alpha_e-\beta_e)}$. In particular $u_e(x;\nu) = 0$ for some $x>0$, if and only if $\frac{b_e}{a_e} = F_\nu(b_e(l_e-x))$ where $F_\nu = \tanh$ for $\nu\leq \nu(e)^*$ and $F_\nu = \tan$ for $\nu > \nu^*(e)$. Since $a_e,b_e>0$ for all $\nu$, and $a_e>b_e$ for $\nu \in [1,\nu^*(e)]$, then there exists $\tilde l_e > l_e$ such that $u_e(\tilde l_e;\nu) = 0$.

Let $\tilde \Gamma$ be the network obtained by extending all leaf edges of $\Gamma$ to their corresponding $\tilde l_e$. Define $\tilde p, \tilde q$, and the operator $\tilde \L$ on $\tilde \Gamma$ via \refp{pq} and \refp{L}. Let $u$ be an eigenfunction corresponding to $\nu_1(\Gamma)$ and $\tilde u$ its extension to $\tilde \Gamma$. Then $\tilde u \in B_D(\tilde \Gamma)$, and $\tilde \L[ \tilde u] = \nuu(\Gamma) \tilde q \tilde u$. Proposition \refp{PropInf} gives $\eta_1(\tilde \Gamma) \leq \nuu(\Gamma)$.
\end{proof}

We now need to introduce some language of oscillation theory for eigenfunctions on networks. To fix ideas, recall that under very general conditions, the eigenfunction of the $k$-th eigenvalue for the one-dimensional Sturm-Liouville problem $\L u = \nu q u$ on $[0,l]$ with $u(0)=u(l) =0$, has exactly $k-1$ zeros in $(0,l)$.  For the analogy of this theory to networks we follow the notation and essential results from \cite{Pokornyi:2004p3313} adapted to the case of tree networks. 

\begin{definition}
An S-zone of a function $f:\Gamma \to \R$ in $\C(\bar\Gamma)$ is a subset $S \subseteq \Gamma$ such that $|f(x)|>0$ for all $x$ in the interior of $S$, and $f = 0$ on $\d S$. 
The pair $(\L,\nu)$ is said to be non-oscillatory on $\Gamma$ if no solution $u\in \Dom(\L)$ to $\L[u] = \nu q u$ has an S-zone in $\Gamma$. In particular, if $(\L,\nu)$ is non-oscillatory on $\Gamma$ there cannot be a sign-constant solution to problem \refp{spectralDir}.
\end{definition}

\begin{theorem}
Let $\nu^*(\Gamma) = \min_{e\in \Gamma} \nu^*(e)$. If $\nu < \nu^*(\Gamma)$, then $(\L,\nu^*(\Gamma))$ is non-oscillatory.
\end{theorem}
\begin{proof}
By theorem 4.2 in \cite{Pokornyi:2004p3313}, it suffices to show that for all $\nu < \nu^*(\Gamma)$, there exists an strictly positive solution $u \in \Dom(\L)$ to $\L[u] = \nu qu$. Let $1<\nu < \nu^*(\Gamma)$, then $0<\alpha_e < \beta_e$ are real numbers for all $e \in \Gamma$. In particular if $B_e>|C_e|$ then $u_e(x;\nu) = B_e e^{\a_e x} + C_e e^{\b_e x} > 0$ for all $x\geq 0$. Assume $e$ is not a leaf edge, and has children edges $e0$ and $e1$. It suffices to show that if $B_e>|C_e|$, there exist choices of $B_{e0},C_{e0},B_{e1},C_{e1}$  such that $u(\cdot,\nu)$ remains positive and satisfies the continuity and matching conditions at node $\e$. Using the continuity of $u$ at $\e$ and the conservation of water equation \refp{consWater}, gives $\derivAD{u}{}(\e) = 0$ if and only if
$$b_{e0}(B_{e0}-C_{e0}) + b_{e1}(B_{e1}-C_{e1}) = b_e(B_e e^{\a_e l_e}-C_e e^{\b_e l_e}).$$
Since $B_e>|C_e|$, then $K:=B_e e^{\a_e l_e}-C_e e^{\b_e l_e} >0$. Let $c_{e0} \in (0,1)$, $c_{e1} := 1-c_{e0}$ and choose 
$$B_{ei} = \frac{1}{2}\left(u(\e;\nu) + c_{ei} \frac{K}{b_{ei}}\right), \quad  
  C_{ei} = u(\e;\nu) - B_{ei} $$
Then $u_e(l_e;\nu) = u_{e0}(0;\nu)=u_{e1}(0;\nu) = u(\e;\nu)$ and $\derivAD{u}{}(\e) = 0$. Moreover
$$|C_{ei}| = \frac{1}{2}\left|u(\e;\nu) - c_{ei} \frac{K}{b_{ei}}\right|<B_{ei},$$
and hence $u_{ei}(\cdot,\nu)$ is positive, $i=0,1$.  
\end{proof}

The proof of theorem \refp{ThmLowBound} follows from the following characterization of the smallest eigenvalue for problem \refp{spectralDir}: $\eta_1(\Gamma)$ is the supremum of real $\eta$ for which $(\L,\eta)$ is non-oscillatory \cite[theorem 5.2.2]{Pokornyi:2004p3313}. Since $\nu^*(\Gamma) = \nu^*(\tilde \Gamma)$, then $(\tilde\L,\nu^*(\Gamma))$ is non-oscillatory on $\tilde \Gamma$, and therefore $\nu^*(\Gamma) < \eta_1(\tilde\Gamma)$. Lemma \refp{LemmaExtend} gives $\nu^*(\Gamma) < \nuu(\Gamma)$.

\section{Discussion}
Corollary \refp{CorStab1D} and theorem \refp{ThmLowBound} give necessary conditions for the imminent extinction of populations in river networks under advection-diffusion. Only a very particular dispersion mechanism is considered: mobile individuals follow realizations of an advection-diffusion process inside $\Gamma$ for an exponential time of mean $\frac{1}{\sigma}$. The boundary conditions are also fixed to be absorbent at the outlet, and reflecting at upstream boundary points. 

Although the stability properties of equation \refp{IDE} may be easily generalized to dispersion kernels with different sets of boundary conditions (see problem \refp{spectralDir} for instance), the methods used here depend heavily on the relation between $\K$ and the Sturm-Liouville operator $\L$, and in particular to the flux condition $\derivAD{f}{}=0$ at internal nodes. Generalization of the persistence conditions to other family of dispersion kernels does not seem feasible. On the other hand, the particular form of $\K$ used here has the conceptual advantage of being derived from microscopical mechanistic considerations, namely those based on the theory of diffusion processes on graphs developed by \cite{Freidlin:1993p4685, Freidlin:2000vn}. 

The analysis of the one-dimensional case of section \refp{Section1d} confirms the results reported on \cite{Lutscher:2005p1894} even when boundary conditions are imposed. Namely, upstream dispersal events in linear advection-diffusion paths, are likely enough to ensure the existence of persistence scenarios, even for high values of the water velocity. 

Theorem \refp{ThmSubtree} hints at the benefits that living on a network brings to organisms susceptible of being washed out of it. In particular, it shows that the population can persist in a river network if upstream refuges are in place, namely subgraphs $\Gamma_e$ of $\Gamma$ where $\rcrit(\Gamma_e)$ is low. The detailed effects of the topology of $\Gamma$ on $\nuu(\Gamma)$ are much harder to track. One reason for this is that the physical variables involved in the dispersal kernel are also dependent on this topology, e.g. the conservation of water equation \refp{consWater}. 

The bounds for $\frac{\rcrit}{\mu}$ in \refp{upbound} might be useful because of their simplicity, but are rather crude. Further work is needed to unravel the relationship between the smallest eigenvalue $\nuu(\Gamma)$ and the physical variables on subgraphs and individual edges of $\Gamma$. Also, river networks in nature exhibit multiple and well-documented scaling properties between channel indexing schemes (e.g. Horton order) and physical variables like channel length, mean velocity, and cross-sectional area \citep[see][]{Rodriguez-Iturbe:1996fk}. Deriving the consequences that such properties have on the estimation of extinction/persistence conditions will prove much useful.

\bibliographystyle{spbasic}  
\bibliography{biblio}

\end{document}